\documentclass[a4paper,UKenglish]{article}

\usepackage{amssymb,amsmath,amsthm}
\usepackage{graphicx}
\usepackage{xspace}
\usepackage{wrapfig}
\usepackage{a4wide}

\usepackage{lineno}

\usepackage{tikz}
\usetikzlibrary{calc,intersections,through,backgrounds}
\usetikzlibrary{arrows, decorations.pathreplacing}

\newcommand{\mytitle}{Fully-Dynamic and Kinetic Conflict-Free Coloring of Intervals %
with Respect to Points\footnote{AvR and MR were supported by JST ERATO Grant Number JPMJER1201, Japan. MdB, AM, and GW were supported by the Netherlands' Organisation for Scientific Research (NWO) under project no.~024.002.003.}\setcounter{footnote}{2}}

\title{\mytitle}

\newcommand{\reals}{\mathbb{R}\xspace}
\newcommand{\Reals}{\reals}
\newcommand{\etal}{\emph{et~al.}\xspace}

\DeclareMathOperator{\bigOmega}{\Omega}

\DeclareMathOperator{\col}{col}

\newcommand{\alg}{{\sc alg}\xspace}
\DeclareMathOperator{\sig}{\mbox{sig}}
\newcommand{\floor}[1]{\left\lfloor #1 \right\rfloor}

\renewcommand{\leq}{\leqslant}
\renewcommand{\geq}{\geqslant}
\newcommand{\tree}{\mathcal{T}}

\newcommand{\parent}{\mathit{parent}}
\newcommand{\eps}{\varepsilon}
\newcommand{\myetal}{\emph{et al.}\xspace}
\newcommand{\myleft}{\mathrm{left}}
\newcommand{\myright}{\mathrm{right}}
\newcommand{\extr}{\mathrm{extr}}
\newcommand{\Sextr}{S_{\extr}}
\newcommand{\lev}{\mathrm{level}}

\newcommand{\mypred}{\mathit{pred}}
\newcommand{\mysucc}{\mathit{succ}}

\theoremstyle{plain}
\newtheorem{theorem}{Theorem}[section]
\newtheorem{lemma}[theorem]{Lemma}

\newcommand{\proofclaim}[1]{\noindent{\bf{}Proof of claim.} #1{}\hfill{\footnotesize$\square$}\vspace{6pt}}
\newcounter{myclaim}
\newcommand{\myclaim}[2]{\vspace{6pt}\noindent{\refstepcounter{myclaim}\bf{}Claim~\themyclaim. }{\it{}#1} \\[6pt]
\proofclaim{#2}}

\definecolor{nicegreen}{rgb}{0,0.7,0.3}
\definecolor{niceyellow}{rgb}{0.9,0.7,0.07}
\definecolor{nicepurple}{rgb}{0.5,0.2,0.8}

\graphicspath{{./figures/}}
\author{Mark de Berg\thanks{Eindhoven University of Technology, the Netherlands.}
        \and
        Tim Leijsen$^{\ddagger}$
        \and
        Aleksandar Markovic$^{\ddagger}$
        \and
        Andr\'e{} van Renssen\thanks{University of Sydney, Australia.}
        \and 
        Marcel Roeloffzen$^{\ddagger}$
        \and
        Gerhard Woeginger\thanks{RWTH Aachen University, Germany.}
        }

\date{}

\begin{document}

\maketitle



\begin{abstract}
We introduce the fully-dynamic conflict-free coloring problem for a set~$S$ of
intervals in~$\Reals^1$ with respect to points, where the goal is
to maintain a conflict-free coloring for~$S$ under insertions
and deletions. A coloring is conflict-free if for each point $p$ contained in some interval, $p$ is contained in an interval whose color is not shared with any other interval containing $p$. We investigate trade-offs between the number
of colors used and the number of intervals that are
recolored upon insertion or deletion of an interval. Our results
include:
\begin{itemize}
\item a lower bound on the number of recolorings as a function of the number
      of colors, which implies that with $O(1)$ recolorings per update the
      worst-case number of colors is $\Omega(\log n/\log\log n)$, and that any strategy
      using $O(1/\eps)$ colors needs $\Omega(\eps n^{\eps})$ recolorings;
\item a coloring strategy that uses $O(\log n)$ colors at the cost of $O(\log n)$ recolorings,
      and another strategy that uses $O(1/\eps)$ colors at the cost of
      $O(n^{\eps}/\eps)$ recolorings;
\item stronger upper and lower bounds for special cases.
\end{itemize}
We also consider the kinetic setting where the intervals move continuously
(but there are no insertions or deletions); here we show how to maintain a coloring
with only four colors at the cost of three recolorings per event and show this is tight.
\end{abstract}


\section{Introduction}
Consider a set~$S$ of fixed base stations that can be used for communication by
mobile clients. Each base station has a transmission range, and a client can potentially
communicate via that base station when it lies within the transmission range. However, when
a client is within reach of several base stations that use the same frequency, the
signals will interfere. Hence, the frequencies of the base stations should be assigned
in such a way that this problem does not arise. Moreover, the number of
used frequencies should not be too large. Even~\etal~\cite{even-cf-03} and Smorodinsky~\cite{thesis-smorodinsky}
introduced conflict-free colorings to model this problem, as follows.
Let~$S$ be a set of disks in the plane, and for a point $q\in \Reals^2$ let
$S(q)\subseteq S$ denote the set of disks containing the point~$q$.
A coloring of the disks in~$S$ is \emph{conflict-free} if,
for any point~$q\in \Reals^2$ with non-empty $S(q)$, the set~$S(q)$ has at least one disk with a
color that is unique among the disks in~$S(q)$.
Even~\etal~\cite{even-cf-03} proved that any set of
$n$~disks in the plane admits a conflict-free coloring with $O(\log n)$ colors,
and this bound is tight in the worst case.

The concept of conflict-free colorings can be generalized and extended in several
ways, giving rise to a host of challenging problems. Below we mention some of them, focussing on the papers most
directly related to our work. A more extensive overview is given
by Smorodinsky~\cite{Smorodinsky2013}. One obvious generalization
is to work with types of regions other than disks.
For instance, Even~\etal~\cite{even-cf-03} showed how to find a coloring with $O(\log n)$ colors for a set of translations of any single centrally symmetric polygon.
Har-Peled and Smorodinsky~\cite{harpeled-cf-05} extended this result to
regions with near-linear union complexity.
One can also consider the dual setting, where one wants to color a given set~$P$
of $n$ points in the plane, such that any disk---or rectangle, or other range from
a given family---contains at least one point with a unique color (if it contains
any point at all). This too was studied by Even~\etal~\cite{even-cf-03} and they
show that this can be done with $O(\log n)$ colors when the ranges are disks
or scaled translations of a single centrally symmetric convex polygon.

The results mentioned above deal with the static setting, in which the
set of objects to be colored is known in advance. This may not always be the
case, leading Fiat~\etal~\cite{fiat-ocf-05} to introduce
the \emph{online} version of the conflict-free coloring problem.
Here the objects to be colored arrive one at a time, and each object must be colored upon arrival.
There is also a lot of work on online coloring of other types (e.g. \cite{kierstead1981extremal,DBLP:reference/algo/Epstein16c}), but we restrict our discussion to conflict-free colorings.
Fiat~\etal~\cite{fiat-ocf-05} show that when coloring points in the plane with respect to disks,
$n$~colors may be needed in the online version. Hence, they turn their attention
to the 1-dimensional problem of online coloring points with respect to intervals.
They prove that this can be done deterministically with $O(\log^2 n)$ colors and randomized with $O(\log n \log\log n)$ colors with high probability.
Later Chen~\cite{chen-ocf-06} gave a randomized algorithm that
uses $O(\log n)$ colors with high probability.
In the same paper, similar results were obtained for conflict-free
colorings of points with respect to halfplanes, unit disks and axis-aligned rectangles
of almost the same size. In these cases the colorings use
$O(\mathrm{polylog}\ n)$ colors with high probability.
Bar-Noy, Cheilaris, and Smorodinsky \cite{Bar-Noy-cf-for-intervals-08}
discussed several versions of the deterministic one-dimensional variant.
Furthermore, Abam~\etal~\cite{Abam2014} studied
the dual version of coloring intervals on a line with respect to points, providing an algorithm that use $O(\log^3 n)$ colors in general or $O(\log n)$ colors when all intervals contain a specific point.
Later, Bar-Noy~\etal~\cite{barnoy-ocf-10} considered the case where
recolorings are allowed for each insertion. They prove that for coloring points
in the plane with respect to halfplanes, one can obtain a coloring with $O(\log n)$
colors in an online setting at the cost of $O(n)$ recolorings in total.
More recent variants include strong conflict-free colorings~\cite{cheilaris-scf-14,hks-cfcms-10},
where we require several unique colors, and conflict-free multicolorings~\cite{barschi-cfm-15},
which allow assigning multiple colors to a point.
Even more variants of online conflict-free colorings
can be found in the survey \cite{Smorodinsky2013}.

\paragraph{Our contributions.}
We introduce a variant of the conflict-free coloring problem where the objects to
be colored arrive and disappear over time. This
\emph{fully-dynamic conflict-free coloring problem} models a scenario
where new base stations may be deployed (to deal with increased capacity demands,
for example) and existing base stations may
break down or be taken out of service (either permanently or temporarily).
We also define the \emph{semi-dynamic conflict-free coloring problem}
as the online variant where recolorings are allowed (or the fully-dynamic
variant without deletions). Note that when we talk about the
\emph{dynamic} variant, we mean \emph{fully-dynamic}.
These natural variants have, to the
best of our knowledge, not been considered so far. It is easy to see that,
unless one maintains a coloring in which any two intersecting objects
have distinct colors, there is always a sequence of deletions that invalidates
a given conflict-free coloring. Hence, recolorings are needed to ensure
that the new coloring is conflict-free. This leads to the question: how many
recolorings are needed to maintain a coloring with a certain number of
colors? We initiate the study of fully-dynamic conflict-free colorings
by considering the problem of coloring intervals
with respect to points. In this variant, we are given a (dynamic) set~$S$ of intervals
in $\Reals^1$, which we want to color such that for any point $q\in \Reals^1$
the set $S(q)$ of intervals containing~$q$ contains an interval with a unique color.
This version of the problem can be used to model the case where the base stations are located along a highway, for instance, and 1-dimensional range and frequency assignment problems have already been studied in various
settings \cite{barnoy-ocf-10, cheilaris-scf-14, fiat-ocf-05}.
Moreover, the lower bounds that we prove hold for the 2-dimensional problem as well.
In the static setting, coloring intervals is rather easy: a simple procedure
yields a conflict-free coloring with three colors. The dynamic version turns out
to be much more challenging.

In Section~\ref{se:lower-bound} we prove lower bounds on the possible tradeoffs
between the number of colors
used and the worst-case number of recolorings per update:
for any algorithm that maintains a conflict-free coloring
on a sequence of~$n$ insertions of intervals
with at most $c(n)$ colors and at most $r(n)$ recolorings per insertion, we
must have~$r(n)>n^{1/(c(n)+1)} / (8c(n))$.
This implies that for $O(1/\eps)$ colors we need $\Omega(\eps n^\eps)$ recolorings per updated,
and with only $O(1)$ recolorings per update we must use $\Omega(\log n / \log\log n)$ colors.

In Section~\ref{se:algorithms} we then present several algorithms that achieve
bounds close to our lower bound. All bounds are worst-case, unless specifically stated otherwise.
First, we present two algorithms for the case where the
interval endpoints come from a universe of size~$U$. One algorithm uses $O(\log U)$ colors
and two recolorings per update; the other uses $O(\log_t U)$
colors and $O(t)$ recolorings per update in the worst case, where $2\leq t \leq U$ is a parameter.
We then extend the second algorithm to an unbounded universe, leading to two
results: we can use $O(\log_t n)$ colors and perform at most $O(t \log_t n)$ recolorings per update
for any fixed $t\geq 2$, or we can use $O(1/\eps)$ colors and $O(n^\eps / \eps)$
recolorings, amortized, for any fixed $\eps > 0$.

Next, in Section~\ref{se:online} we consider the online setting, where intervals arrive one by one
and recolorings are not allowed. Abam \etal~\cite{Abam2014} observed that for this
setting there is a lower bound of $\log n$ on the number of required colors. 
They also provide an upper bound of $O(\log^3 n)$ for general instances, as well as an upper bound of $\log n$ when all intervals are nested and contain the origin. We provide an upper bound of $\log n$ on nested intervals even if not all of them contain the origin. We also use the result of Bar-Noy \etal~\cite{barnoy-ocf-10} to provide a non-deterministic algorithm on any instance (non-nested) with $O(\log n)$ colors with high probability.

Finally, in Section~\ref{se:kinetic} we turn our attention to
\emph{kinetic conflict-free colorings}. Here the intervals do not appear or disappear,
but their endpoints move continuously on the real line. At each \emph{event} where
two endpoints of different intervals cross each other, the coloring may need to
be adapted so that it stays conflict-free. One way to handle this is to
delete the two intervals involved in the event, and re-insert them with the
new endpoint order. We show that a specialized approach is much more efficient:
we show how to maintain a conflict-free coloring with four colors at the cost of
three recolorings per event. We also show that on average $\Theta(1)$
recolorings per event are needed in the worst case when using only four colors.

\paragraph{Preliminaries: the chain method.}
We first present the so called \emph{chain method}
that colors intervals statically with three colors.
This method is used later, but we present it here
to give the reader a sense of CF-colorings.
Consider the intervals~$I_1,\ldots,I_n$ and suppose
without loss of generality that their union
is an interval.
Let $I_i$ be the
    interval with the leftmost left endpoint in this component.
    We color $I_i$ red. Then, of all intervals whose left endpoint lies inside~$I_i$
    we select the interval $I_j$ with the rightmost right endpoint and color it blue,
    of all intervals whose left endpoint lies inside~$I_j$
    we select the interval~$I_k$ sticking out furthest to the right and color it red,
    and so on. This way we create a chain of intervals that are alternatingly
    colored red and blue. We then color all
    the remaining intervals using the dummy color,
    see Fig.~\ref{fi:chain}.
    \begin{figure}[bt]
    \begin{center}
    \includegraphics{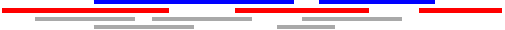}
    \end{center}
    \caption{Coloring statically intervals using the chain method.
             Gray is the dummy color.}
    \label{fi:chain}
    \end{figure}

\section{Lower bounds for semi-dynamic conflict-free colorings}\label{se:lower-bound}
In this section we present lower bounds on the semi-dynamic (insertion
only) conflict-free coloring problem for intervals. More precisely,
we present lower bounds on the number of recolorings necessary to
guarantee a given upper bound on the number of colors. We prove
a general lower bound and a stronger bound for so-called local algorithms.
The general lower bound uses a construction where the choice
of segments to be added depends on the colors of
the segments already inserted. This adaptive construction is also valid for
randomized algorithms, but it does not give a lower bound on the expected behavior.

\begin{theorem}\label{thm_lowerbound}
Let \alg be a deterministic algorithm for the semi-dynamic conflict-free
coloring of intervals. Suppose that on any sequence of $n>0$~insertions,
\alg uses at most $c(n)$~colors and $r(n)$~recolorings per insertion,
where $r(n)>0$. Then
$r(n) > n^{1/(c(n)+1)} / (8c(n))$.
\end{theorem}
\begin{proof}
We first fix a value for~$n$ and define~$c:=c(n)$ and~$r:=r(n)$.
Our construction will proceed in rounds. In the~$i$-th
round we insert a set~$R_i$ of~$n_i$ disjoint intervals---which
intervals we insert depends on the current coloring provided
by \alg. After $R_i$ has been inserted (and colored by \alg),
we choose one of the colors used by \alg for $R_i$ to be the
\emph{designated color} for the $i$-th round. We denote this
designated color by~$c_i$. We will argue that in each round we can pick a
different designated color, so that the number of rounds, $\rho$, is a lower bound
on the number of colors used by \alg. We then prove a lower bound on $\rho$
in terms of~$n, c$, and~$r$, and  derive the theorem from the inequality~$\rho \leqslant c$.
\medskip

To describe our construction more precisely, we need to introduce
some notation and terminology. Let~$R_i := \{I_1,\ldots,I_{n_i}\}$,
where the intervals are numbered from left to right. (Recall that
the intervals in~$R_i$ are disjoint.) To each interval~$I=I_j$
we associate the set~$I^e:=(a,b)$,
where~$a$ is the right endpoint of~$I$,
and~$b$ is the left endpoint of~$I_{j+1}$ if~$j<n_i$ and~$+\infty$ if~$j=n_i$, that is, $I^e$ represents the empty space to the right of $I$.
We call~$(I,I^e)$ an \emph{$i$-brick}. We define the color of
a brick~$(I, I^e)$ to be the color of~$I$, and we say a point
or an interval is contained in this brick if it is contained in~$I\cup I^e$.
Recall that each round~$R_i$ has a designated color~$c_i$.
We say that an $i$-brick~$B:=(I,I^e)$ is \emph{living} if:
\begin{itemize}
\item $I$ has the designated color~$c_i$;
\item if~$i>1$ then both~$I$ and~$I^e$ contain living $(i-1)$-bricks.
\end{itemize}
A brick that is not alive is called \emph{dead} and an event such as a recoloring that causes a brick to become dead is said to \emph{kill} the brick.
By recoloring an interval~$I$, \alg can kill the brick $B=(I,I^e)$
and the death of $B$ may cause some bricks containing~$B$ to be killed as well.

We can now describe how we generate the set~$R_i$ of intervals we
insert in the $i$-th round and how we pick the designated colors.
(Note that the designated color of a round is fixed once it is picked;
it is not updated when recolorings occur.)
We denote by~$R_i^*$ the subset of intervals~$I\in R_i$ such that
$(I,I^e)$ is a living $i$-brick. Note that $R_i^*$ can be defined only
after the $i$-th round, when we have picked the designated color~$c_i$.
\begin{enumerate}
\item The set~$R_1$ contains the~$\frac{n}{2}$
    intervals~$[0,1],[2,3],\ldots,[n-2,n-1]$, and the designated
    color~$c_1$ of the first round is the color used most often in the
    coloring produced by \alg after insertion of the last interval in~$R_1$.
\item To generate~$R_i$ for~$i>1$, we proceed as follows.
    Partition~$R^*_{i-1}$ into groups
    of~$4r$ consecutive intervals. (If~$|R^*_{i-1}|$ is not a multiple
    of~$4r$, the final group will be smaller than~$4r$. This group will be ignored.)
    For each group $G := I_1,\ldots,I_{4r}$ we put an interval $I_G$ into~$R_i$,
    which starts at the left endpoint of~$I_1$ and ends slightly before the
    left endpoint of~$I_{2r+1}$; see Fig.~\ref{fi:bricks} for an illustration.
    \begin{figure}[bt]
    \begin{center}
    \includegraphics{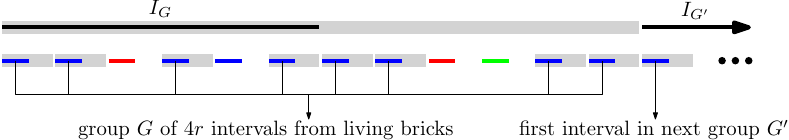}
    \end{center}
    \caption{Example of how the intervals are created when $r=2$.
    The designated color $c_{i-1}$ is blue, and the grey rectangles around them
    indicate living $(i-1)$-bricks. The grey rectangle around $I_G$ indicates the
    brick $(I_G,I_G^e)$. Note that $I_{G'}$ extends further to the right.}
    \label{fi:bricks}
    \end{figure}

    The designated color~$c_i$ is picked as follows. Consider the coloring
    after the last interval of~$R_i$ has been inserted, and let~$C(i)$ be
    the set of colors assigned by \alg to intervals in~$R_i$ and that are
    not a designated color from a previous round---we argue below
    that~$C(i)\neq \emptyset$. Then we pick~$c_i$ as the color from~$C(i)$
    that maximizes the number of living $i$-bricks.
\end{enumerate}
We continue generating sets~$R_i$ in this manner until~$|R^*_{i}| <4r$,
at which point the construction finishes.
Below we prove that in each round \alg must introduce a new designated color,
and we prove a lower bound on the number of rounds in the construction.

\myclaim{%
Let $B=(I, I^e)$ be a living $i$-brick. Then for any $j\in\{1,\ldots,i\}$
there is a point~$q_j\in I\cup I^e$ that is contained in a single interval
of color~$c_j$ and in no other interval from $\bigcup_{\ell=1}^{i-1} R_\ell$.
Moreover, there is a point~$q_j\in I\cup I^e$ not contained in any interval
from~$\bigcup_{\ell=1}^{i-1} R_\ell$.%
}
{%
We prove this by induction on~$i$. For $i=1$ the statement is trivially true,
so suppose $i>1$. By definition, both $I$ and $I^e$ contain living
$(i-1)$-bricks, $\overline{B}$ and $\overline{B}^e$.
Using the induction hypothesis we can now select a point~$q_j$ with the desired
properties: for $j=i$ we use that $\overline{B}$ contains a point
that is not contained in any interval from $\bigcup_{\ell=1}^{i-1} R_\ell$,
for $j< i$  we use that $\overline{B}^e$ contains a point in an interval
of color~$c_j$ and in no other interval from $\bigcup_{\ell=1}^{i-1} R_\ell$,
and to find a point not contained in any interval from $\bigcup_{\ell=1}^{i-1} R_\ell$
we can also use~$\overline{B}^e$.%
}

Now consider the situation after the $i$-th round, but before we have chosen the
designated color~$c_i$.  We say that a color~$c$ is \emph{eligible} (to become $c_i$)
if $c \neq c_1,\ldots,c_{i-1}$, and we say that an $i$-brick $(I,I^e)$ is eligible if
its color is eligible and $I$ and $I^e$ both contain living $(i-1)$-bricks (if $i>1$).
Note that due to some recolorings, some of the newly inserted intervals
might not contain any living brick and hence can never be living no matter
the designated color; the next claim shows that at most half intervals
inserted this round are eligible.

\myclaim{%
Immediately after the $i$-th round, at least half of the $i$-bricks are eligible.%
}
{%
Consider an $i$-brick $(I,I^e)$. At the beginning of the $i$-th round, before we have actually inserted
the intervals from $R_i$, both the interval~$I$ and its empty space~$I^e$ contain~$2r$ living
$(i-1)$-bricks. As the intervals from $R_i$ are inserted, \alg may recolor certain intervals from
$R_1\cup\ldots\cup R_{i-1}$, thereby killing some of these $(i-1)$-bricks.
Now suppose that \alg recolored at most $2r-1$ of the intervals from $R_1\cup\ldots\cup R_{i-1}$
that are contained in $I\cup I^e$. Then both $I$ and $I^e$ still contain a living $(i-1)$-brick.
By the previous claim and the definition of a conflict-free coloring, this implies \alg cannot use any of the colors~$c_j$ with $j<i$
for~$I$. Hence, the color of $I$ is eligible and the $i$-brick $(I,I^e)$ is eligible as well.

It remains to observe that \alg can do at most $r n_i$ recolorings during the $i$-th round.
We just argued that to prevent an $i$-brick from becoming eligible, \alg must do at least
$2r$~recolorings inside that brick. Hence, \alg can prevent at most half of the $i$-bricks from
becoming eligible.%
}

Recall that after the $i$-th round we pick the designated color~$c_i$ that maximizes the
number of living $i$-bricks. In other words, $c_i$ is chosen to maximize $|R^*_i|$.
Next we prove a lower bound on this number. Recall that $\rho$ denotes the number of rounds.

\myclaim{%
For all $1\leq i\leq \rho$ we have $|R^*_i| \geq n_1/(8rc)^i -1$.%
}
{%
Since \alg can use at most $c$ colors, we have $|R^*_1| \geq n_1/c$.
Moreover, for $i>1$ the number of intervals we insert is $\floor{|R^*_{i-1}|/4r}$.
By the previous claim at least half of these are eligible.
The eligible intervals have at most~$c$ different colors,
so if we choose~$c_i$ to be the most common color among them we see that
$|R^*_i| \geq \floor{|R^*_{i-1}|/4r}/2c$. We thus obtain the following
recurrence:
\begin{equation}\label{eq-rici}
|R^*_i| \geqslant \left\lbrace \begin{array}{l l}
\dfrac{\floor{|R^*_{i-1}|/4r}}{2c} & \text{if } i > 1, \\[10pt]
\dfrac{n_1}{c} & \text{if } i=1.
\end{array} \right.
\end{equation}
We can now prove the result using induction.
\[
|R^*_i| \ \ \geq \ \ \frac{\floor{|R^*_{i-1}|/4r}}{2c}
        \ \ \geq \ \ \frac{1}{2c}\cdot\left(  \left( \frac{n_1}{(8rc)^{i-1}} -1 \right) / 4r  -1 \right)
        \ \  > \ \  \frac{n_1}{(8rc)^{i}} - 1.
\]%
}

Finally we can derive the desired relation between~$n, c$, and~$r$.
Since~$n_1=n/2$ and $n_{i+1}<n_i/2$ for all~$i=1,\ldots,\rho-1$,
the total number of insertions is less than~$n$.
The construction finishes when $|R^*_i|<4r$. Hence, $\rho$,
the total number of rounds, must be such that
$n/(2(8rc)^{\rho}) - 1 \leq |R^*_{\rho}| < 4r$,
which implies $\rho > \log_{8rc} (n/(8r+2)) > \log_{8rc}n-1$.
The number of colors used by \alg is at least $\rho$,
since every round has a different designated color.
Thus $c>\log_{8rc}n-1$ and so $n\leq (8rc)^{c+1}$, from which the
theorem follows.
\end{proof}
Two interesting special cases of the theorem are the following:
with $r=O(1)$ we will have $c=\bigOmega\left(\log n/\log \log n \right)$,
and for~$c=O(1/\eps)$ (for some small fixed~$\eps>0$) we
need~$r=\bigOmega\left(\eps n^\varepsilon \right)$.
Note that the theorem requires $r>0$. Obviously the
$\bigOmega\left(\log n/\log \log n \right)$ lower bound on $c$
that we get for $r=1$ also holds for~$r=0$. For the special case of $r=0$---this
is the standard online version of the problem---we can
prove a stronger result, however: here we need at least $\lfloor \log n\rfloor +1$
colors. 
This bound even holds for a nested
set of intervals, that is, a set~$S$ such that
$I\subset I'$, $I'\subset I$, or~$I\cap I'=\emptyset$ for any two intervals $I,I'\in S$.
We also show in Section~\ref{se:online} that a greedy algorithm achieves this bound
for nested intervals.

\subsection{Local algorithms}
We now prove a stronger lower bound for so-called local algorithms.
Intuitively, these are deterministic algorithms where the color assigned to a newly
inserted interval~$I$ only depends on the structure and the coloring of the
connected component where~$I$ is inserted---hence the name \emph{local}.
More precisely, local algorithms are defined as follows.

Suppose we insert an interval $I$ into a set $S$ of intervals
that have already been colored. The union of the set $S\cup \{I\}$ consists
of one or more connected components. We define $S(I)\subseteq S$ to be the
set of intervals from $S$ that are in the same connected component as~$I$.
(In other words, if we consider the interval graph induced by $S\cup \{I\}$
then the intervals in $S(I)$ form a connected component with~$I$.)
Order the intervals in $S(I)\cup\{I\}$ from left to right according to
their left endpoint, and then assign to every interval its rank in this
ordering as its label. (Here we assume that all endpoints of the intervals
in $S(I)\subseteq S$ are distinct. It suffices to prove our lower bound
for this restricted case.) Based on this labelling we define a signature
for $S(I)\cup\{I\}$ as follows. Let $\lambda_1,\ldots,\lambda_k$, where
$k:= |S(I)|+1$, be the sequence of labels obtained by ordering the intervals
from left to right according to their right endpoint. Furthermore,
let~$c_i$ be the color of the interval labeled~$\lambda_i$, where
$c_i=\mbox{\sc nil}$ if the interval labeled~$i$ has not yet been colored.
Then we define the \emph{signature} of $S(I)\cup I$ to be the sequence
$\sig(I) := \langle \lambda_1,\ldots,\lambda_k,c_1,\ldots,c_k\rangle$;
see Fig.~\ref{fi:signature}.
\begin{figure}[bt]
    \begin{center}
    \includegraphics{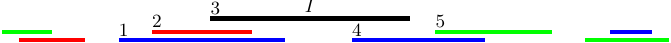}
    \end{center}
    \caption{Example of a signature. The set~$S(I)$ contains the segments labeled 1,2,4,5. The signature of $I$ is $\langle 2,1,3,4,5,\mbox{red},\mbox{blue},\mbox{\sc nil},\mbox{blue},\mbox{green}\rangle$.}
    \label{fi:signature}
    \end{figure}

We now define a semi-dynamic algorithm \alg to be \emph{local} if upon
insertion of an interval~$I$ the following holds:
(i) \alg only performs recoloring in $S(I)$, and
(ii) the color assigned to $I$ and the recolorings in $S(I)$
are uniquely determined by~$\sig(I)$, that is, the algorithm is deterministic with respect to~$\sig(I)$.
Note that randomized algorithms are not local.
Also note that when a set of disjoint intervals is inserted into an initially empty set, a local algorithm will give each interval the same color. Thus the algorithms we present in the next section are not local, as the coloring depends on a global tree structure.

To strengthen Theorem~\ref{thm_lowerbound} for the case of
local algorithms, it suffices to observe that the intervals
inserted in the same round must all receive the same color.
Hence, the factors~$c$
in the denominators of Inequality~(\ref{eq-rici}) disappear,
leading to the theorem below. Note that for~$r(n)=O(1)$,
we now get the lower bound $c(n)=\bigOmega(\log n)$.
\begin{theorem}\label{thm_lowerbound_local}
Let \alg be a local algorithm for the semi-dynamic conflict-free
coloring of intervals. Suppose that on any sequence of $n>0$~insertions,
\alg uses at most $c(n)$~colors and $r(n)$~recolorings
per insertion, where $r(n)>0$. Then
$r(n) \geqslant n^{1/(c(n)+2)} - 2$.
\end{theorem}
\begin{proof}
We first fix a value for~$n$ and define~$c:=c(n)$ and~$r:=r(n)$.
Similar to the previous construction, this construction also proceeds in rounds. In the~$i$-th
round we insert a set~$R_i$ of~$n_i$ disjoint intervals.
We will argue that in each round \alg picks a
different color, so that the number of rounds, $\rho$, is a lower bound
on the number of colors used by \alg. We then prove a lower bound on $\rho$
in terms of~$n, c$, and~$r$, and  derive the theorem from the inequality~$\rho \leqslant c$.

In the $i$-th round, we insert the following intervals:
\begin{enumerate}
\item The set~$R_1$ contains the~$\frac{n}{2}$
    intervals~$[0,1],[2,3],\ldots,[2n-2,2n-1]$. Since all intervals are disjoint, the local information \alg can use for each interval consists only of that single interval. Hence, since \alg is deterministic, it colors all intervals in $R_1$ with the same color.
    \item To generate~$R_i$ for~$i>1$, we proceed as follows.
    Partition~$R_{i-1}$ into groups
    of~$r+2$ consecutive intervals. (If~$|R_{i-1}|$ is not a multiple
    of~$r+2$, the final group will be smaller than~$r+2$. This group will be ignored.)
    For each group $G := I_1,\ldots,I_{r+2}$ we put an interval $I_G$ into~$R_i$,
    which starts at the left endpoint of~$I_1$ and ends slightly before the
    left endpoint of~$I_{r+2}$. Since local information for \alg for each of these
    intervals is the same and \alg is deterministic, it colors all
    intervals of $R_i$ using the same color and the output
    signature of $S(I)$ is the same for all $I\in R_i$.
\end{enumerate}
We continue generating sets~$R_i$ in this manner until~$|R^*_{i}| < r+1$, at which point the construction finishes. Note that if $R_{\rho-1}$ consists of exactly $r+1$ intervals, we cannot apply the above construction as it is. In this case, $R_{\rho}$ consists of a single interval which starts at the left endpoint of~$I_1$ and ends at~$2n-1$, that is the new interval contains every other
interval.

By construction, for each interval $I$ in $R_i$, for $i \in \{1, \ldots, \rho\}$, and for each $j \in \{1,\ldots,i\}$ there is a point~$q$ that is contained in a single interval of $R_j$ and in no other interval from $R_1\cup\cdots \cup R_i$. Moreover, the rightmost point of $I$ is not contained in any interval from~$R_1\cup\cdots \cup R_i$. Furthermore, since for each newly inserted interval $I$ \alg can only recolor $r$ of the intervals that $I$ intersects, there exists a set of intervals $I_1, \cdots, I_{i-1}$ such that for every $I_j$ in this set, $I_j \in R_j$, $I_j$ is contained in $I_{j+1}$ ($I_i = I$), and $I_j$ is not recolored by \alg. Hence, these properties remain valid during the insertion process.

\myclaim{%
For all $1\leq i\leq \rho$ we have $|R_i| \geq n/(r+2)^i -2$.%
}
{%
Since \alg colors all intervals of $R_1$ using the same color, we have $|R_1| = n/2$.
Moreover, since for $i>1$ the number of intervals we insert is $\floor{|R_{i-1}|/(r+2)}$, we have that
$|R_i| \geq \floor{|R_{i-1}|/(r+2)}$.
We thus obtain the following
recurrence:
\begin{equation*}
|R_i| \geqslant \left\lbrace \begin{array}{l l}
\floor{\dfrac{|R_{i-1}|}{r+2}} & \text{if } i > 1, \\[20pt]
\dfrac{n}{2} & \text{if } i=1.
\end{array} \right.
\end{equation*}
We can now prove the result using induction.
\[
\begin{array}{lll}
|R_i| & \geq & \floor{|R_{i-1}|/(r+2)} \\[2mm]
        & \geq & \floor{ \left( \frac{n}{(r+2)^{i-1}} -2 \right) / (r+2)} \\[2mm]
        & \geq & \left(  \left( \frac{n}{(r+2)^{i-1}} -2 \right) / (r+2)  -1 \right) \\[2mm]
        & = & \frac{n}{(r+2)^{i}} - \frac{2}{r+2}  - 1 \\[2mm]
        & > & \frac{n}{(r+2)^{i}} - 2.
\end{array}
\]%
}

Finally we can derive the desired relation between~$n, c$, and~$r$.
Since~$|R_1|=n/2$ and $|R_{i+1}|<|R_i|/2$ for all~$i=1,\ldots,\rho-1$,
the total number of insertions is less than~$n$.
The construction finishes when $|R_i|<r+1$. Hence, $\rho$,
the total number of rounds, must be such that
\[
 \frac{n/2}{(r+2)^{\rho}} - 2 \leq |R_{\rho}| < r+1,
\]
which implies $\rho > \log_{r+2} (n/(2r+6)) > \log_{r+2}n-2$.
The number of colors used by \alg is at least $\rho$,
since every round uses a different color.
Thus $c>\log_{r+2}n-2$ and so $n\leq (r+2)^{c+2}$, from which the
theorem follows.
\end{proof}

\section{Upper bounds for fully-dynamic conflict-free colorings}\label{se:algorithms}
Next we present algorithms to maintain a conflict-free coloring
for a set $S$ of intervals under insertions and deletions. The algorithms
use the same structure, which we describe first. From now on, we
use $n$ to denote the current number of intervals in~$S$.
\medskip

Let $P$ be the set of $2n$ endpoints of the intervals in~$S$.
(To simplify the presentation we assume that all endpoints are distinct,
but the solution is easily adapted to the general case.)
We will maintain a B-tree on the set~$P$. A B-tree of minimum degree~$t$ on a set of
points in $\reals^1$ is a multi-way search tree in which each internal node
has between $t$ and $2t$ children (except the root, which may have
fewer children) and all leaves are at the same level; see the book by
Cormen~\myetal~\cite[Chapter~18]{clrs-ia-09} for details. Thus each internal
or leaf node stores between $t-1$ and $2t-1$ points from $P$
(again, the root may store fewer points). We denote the set of points
stored in a node $v$ by $P(v) := \{p_1(v),\ldots,p_{n_v}(v)\}$,
where $n_v:= |P(v)|$ and the points are numbered from left to right.
For an internal node $v$ we denote the $i$-th subtree of~$v$, where
$1\leq i\leq n_v+1$, by $\tree_i(v)$. Note that the search-tree property
guarantees that all points in $\tree_i(v)$ lie in the range $(p_{i-1}(v),p_{i}(v))$, where $p_0 = -\infty$ and $p_{n_v+1} = \infty$.

We now associate each interval $I\in S$ to the highest
node~$v$ such that $I$ contains at least one of the points in~$P(v)$,
either in its interior or as one of its endpoints. Thus our
structure is essentially an interval tree~\cite[Chapter~10]{bcko-cgaa-08}
but with a B-tree as underlying tree structure.  We denote the
set of intervals associated to a node~$v$ by $S(v)$.
Note that if $\lev(v)=\lev(w)=\ell$, for some nodes $v\neq w$,
and $I\in S(v)$ and $I'\in S(w)$,
then $I$ and $I'$ are separated by a point $p_i(z)$ of some node~$z$
at level~$m < \ell$. Hence, $I\cap I'=\emptyset$.

We partition $S(v)$ into subsets $S_1(v),\ldots,S_{n_v}(v)$ such that
$S_i(v)$ contains all intervals $I\in S(v)$ for which
$p_i(v)$ is the leftmost point from $P(v)$ contained in~$I$.
From each subset $S_i(v)$ we pick at most two \emph{extreme intervals}:
the \emph{left-extreme interval}~$I_{i,\myleft}(v)$ is the interval
from $S_i(v)$ with the leftmost left endpoint, and the
\emph{right-extreme interval}~$I_{i,\myright}(v)$ is
the interval from~$S_i(v)$ with the rightmost right endpoint.
Since all intervals from~$S_i(v)$ contain the point~$p_i(v)$,
every interval from~$S_i(v)$ is contained in
$I_{i,\myleft}(v) \cup I_{i,\myright}(v)$. Note that it
may happen that $I_{i,\myleft}(v) = I_{i,\myright}(v)$.
Finally, we define
$\Sextr(v) := \bigcup_{i=1}^{n_v} \{I_{i,\myleft}(v), I_{i,\myright}(v) \}$
to be the set of all extreme intervals at~$v$.

Our two coloring algorithms both maintain a coloring with the following properties.
\begin{description}
\item[(A.1)] For each level~$\ell$ of the tree~$\tree$, there is a set $C(\ell)$ of colors
      such that these color sets are disjoint between different levels.
\item[(A.2)] For each node $v$ at level $\ell$ in $\tree$, the intervals from $\Sextr(v)$
      are colored locally conflict-free using colors from $C(\ell)$ and
      a universal dummy color. Here
      \emph{locally conflict-free} means that the coloring of~$\Sextr(v)$
      is conflict-free if we ignore all other intervals.
      \item[(A.3)] All non-extreme intervals receive a universal
      \emph{dummy color},
      which is distinct from any of the other colors used,
      that is, the dummy color is not in $C(\ell)$ for
      any level $\ell$.
\end{description}
The two coloring algorithms that we present differ in the size
of the sets~$C(\ell)$ and in which local coloring algorithm is used for the
sets~$\Sextr(v)$.
It is not hard to show that the
properties above guarantee a conflict-free coloring.
\begin{lemma}\label{le:interval-tree}
Any coloring with properties (A.1)--(A.3) is conflict free and uses
at most $1+\sum_{\ell} |C(\ell)|$ colors.
\end{lemma}
Next we describe two algorithms based on this general framework:
one for the easy case where the interval endpoints come from
a finite universe, and one for the general case.

\paragraph{Solutions for a polynomially-bounded universe.}
The framework above uses a B-tree on the interval endpoints.
If the interval endpoints come from a universe of size~$U$---for concreteness,
assume the endpoints are integers in the range $0,\ldots,U-1$---then
we can use a B-tree on the set~$\{0,\ldots,U -1\}$. Thus
the B-tree structure never has to be changed. 
\begin{theorem}\label{th:bounded-universe}
Let $S$ be a dynamic set of intervals whose endpoints are integers in~$\{0,\ldots,U-1\}$.
\begin{enumerate}
\item[(i)] We can maintain a conflict-free coloring
    on~$S$ that uses $O(\log U)$ colors and that performs at most two
    recolorings per insertion and deletion.
\item[(ii)] For any $t$ with $2\leq t\leq U$, we can maintain a conflict-free coloring
    on~$S$ that uses $O(\log_t U)$ colors and performs $O(t)$
    recolorings per insertion and deletion.
\end{enumerate}
\end{theorem}
\begin{proof}
For both results we use the general framework described above.
\begin{enumerate}
\item[(i)]
    We pick $|C(\ell)|=4t-2$ for all levels~$\ell$. Since a B-tree
    of minimum degree~$t$ on $U$~points has height~$O(\log_t U)$, the total number of colors
    is $O(t\log_t U)$. With $4t-2$~colors, we can give each interval in $\Sextr(v)$
    its own color, so the coloring is locally conflict free. To get $O(\log U)$
    colors we now pick~$t=2$.
    To insert an interval~$I$, we find the set~$S_i(v)$ into which $I$ should be inserted,
    and check if $I$ should replace the current left-extreme and/or right-extreme
    interval in~$S_i(v)$. Thus an insertion requires at most two recolorings.
    Deletions can be handled in a similar way: if $I$ is an extreme interval,
    then we have to recolor at most two intervals that become extreme
    after the deletion of~$I$.
\item[(ii)]
    We pick $|C(\ell)|=2$ for all levels~$\ell$, yielding $O(\log_t U)$ colors in total.
    We now color $\Sextr(v)$ using the chain method
    with~$C(\lev(v))$ and the global dummy color.
    After coloring every connected component of $\Sextr(v)$ in this manner
    we have a locally conflict-free coloring of~$\Sextr(v)$. Insertions and deletions
    into $S(v)$ are simply handled by updating~$\Sextr(v)$ and recoloring~$\Sextr(v)$
    from scratch. Since at most two intervals start or stop being extreme
    and $|\Sextr(v)|\leq 4t-2$, we use $O(t)$ recolorings.
\end{enumerate}
\vspace{-1.5\baselineskip}
\end{proof}

When~$U$ is polynomially bounded in $n$---that is, $U=O(n^k)$ for some
constant~$k$---this gives very efficient coloring schemes.
In particular, we can then get $O(\log n)$ colors with at most two recolorings
using method~(i), and we can get $O(1/\eps)$ colors with $O(n^{\eps})$ recolorings
(for any fixed~$\eps>0$) by setting $t=U^{\eps/k}$ in method~(ii).

Note finally that we do not need to explicitly store the whole
tree as it is enough to compute the location of any node
when needed, yielding a linear space complexity.

\paragraph{A general solution.}
If the interval endpoints do not come from a bounded universe then we cannot
use a fixed tree structure. Next we explain how to deal with this
when we apply the method from Theorem~\ref{th:bounded-universe}(ii),
which colors the sets $\Sextr(v)$ using the chain method:
we take the interval with the leftmost left endpoint, and color it blue.
Then, among all intervals whose left endpoint lies in the blue interval,
we pick the one with the rightmost right endpoint and color it red.
We then repeat this process, alternating between blue and red,
until we reach the rightmost interval.
Finally, we color all uncolored intervals grey.
\medskip

Suppose we want to insert a new interval~$I$ into the set~$S$. We first insert
the two endpoints of $I$ into the B-tree~$\tree$. Inserting an endpoint~$p$
can be done in a standard manner. The basic operations for an insertion are
(i) to split a full node and (ii) to insert a point into a non-full leaf node.

Splitting a full node~$v$ (that is, a node with $2t-1$ points) is done by moving
the median point into the parent of~$v$, creating a node containing the
$t-1$~points to the left of the median and another node containing the
$t-1$~points to the right of the median. Note that this means that
some intervals from $S(v)$ may have to be moved to $S(\parent(v))$.
Thus splitting a node~$v$ involves recoloring intervals in
$S(v)$ and $S(\parent(v))$. Observe that an interval only needs to be recolored
if it was extreme before or after the change. Hence, we recolor
$O(t)$ intervals when we split a node~$v$.

Since an insertion splits only nodes on a root-to-leaf
path and the depth of~$\tree$ is $O(\log_t n)$, the total number of
recolorings due to node splitting is $O(t \log_t n)$.
Moreover, inserting a point into a non-full leaf node
only takes $O(t)$ recolorings. We conclude that an insertion
performs $O(t \log_t n)$ recolorings in total.
For deletions the argument is similar.
Since recoloring at a single node induces $O(t)$ recolorings,
the total number of recolorings is~$O(t \log_t n)$.

\begin{theorem}\label{th:general-upper-bound}
Let $S$ be a dynamic set of intervals.
\begin{enumerate}
\item[(i)] For any fixed $t\geq 2$ we can maintain a conflict-free coloring on~$S$
  that uses $O(\log_t n)$ colors and that performs $O(t \log_t n)$ recolorings
  per insertion and deletion, where $n$ is the current number of intervals in~$S$.
  In particular, we can maintain a conflict-free coloring
  with $O(\log n)$ colors using $O(\log n)$ recolorings per update.
\item[(ii)] For any fixed $\eps>0$ we can maintain a conflict-free coloring on~$S$
  that uses $O(1/\eps)$ colors and that performs $O(n^{\eps}/\eps)$ recolorings
  per insertion or deletion. The bound on the number of recolorings is amortized.
\end{enumerate}
\end{theorem}
The idea behind part (ii) is to use a $t$ with $n^{\eps}/2\leq t\leq 2n^{\eps}$.  This causes the bound in (ii) to be amortized,
since now we need to change $t$ when $n$ has halved or doubled.

We have not been able to efficiently generalize the first method of
Theorem~\ref{th:bounded-universe} to an unbounded
universe. The problem is that splitting a node~$v$
may require many intervals in $\Sextr(v)$ to be recolored,
since many intervals may be moved to $\parent(v)$. Hence, the method
would use the same number of recolorings as the chain method, but more colors.

\paragraph{Bounded-length intervals.}
Next we present a simple method that allows us to improve the bounds when the
intervals have length between~$1$ and $L$ for some constant $L>1$.
\begin{theorem}\label{th:small-intervals}
Let $S$ be a dynamic set of intervals with lengths in the range $[1,L)$ for
some fixed $L>1$.
Suppose we have a dynamic conflict-free coloring algorithm for a general set of
intervals that uses at most $c(n)$~colors and at most $r(n)$~recolorings for any
insertion or deletion. Then we can obtain a  dynamic conflict-free coloring algorithm
on $S$ that uses at most $2\cdot c(2L)+1$ colors and at most $2\cdot r(2L)+1$ recolorings
for any insertion or deletion.
\end{theorem}
\begin{proof}
For each pair of integers $i,j$ so that $i\geq 0$ and $0\leq j<L$ let $x_{i,j} := iL+j$, and for each $i\geq 0$ define  $X_i := \{x_{i,j} : 0\leq j<L\}$.
Note that the point sets $X_i$ form a partition of the non-negative integer points in $\Reals^1$
into subsets of $L$~consecutive points.
We assign each interval $I\in S$ to the leftmost point~$x_{i,j}$ contained in it, and we
denote the subset of intervals assigned to a point in $X_i$ by $S_i$.
We will color each subset~$S_i$ separately. Since an interval in~$S_i$
can intersect only intervals in~$S_{i-1}\cup S_{i}\cup S_{i+1}$,
we can use the same color set for all $S_i$ with $i$ even, and the same color
set for all $S_i$ with $i$ odd. It remains to argue that we can maintain a coloring for~$S_i$ using
$c(2L)$ colors (in addition to a global dummy color) and $2\cdot r(2L)+1$ recolorings.

Let $\Sextr(x_{i,j})$ contain the two extreme intervals assigned to~$x_{i,j}$:
the interval sticking out furthest to the left and, among the
remaining intervals, the one sticking out furthest to the right.
We give each non-extreme interval the dummy color
and maintain a conflict-free coloring for~$\Sextr(i) := \bigcup_j \Sextr(x_{i,j})$.
When we insert a new interval into $S_i$ it replaces at most one
extreme interval, so we do at most one insertion and one deletion in~$\Sextr(i)$.
Thus we do at most $2 \cdot r(2L)+1$ recolorings. (The +1 is because the interval
that becomes non-extreme is recolored to the dummy color.)
For deletions the argument is similar.
\end{proof}
For instance, by applying Theorem~\ref{th:general-upper-bound}(i) we can maintain
a coloring with $O(\log L)$ colors and $O(\log L)$ recolorings. We can also
plug in a trivial dynamic algorithm with $c(n)=n$ and $r(n)=0$ to obtain $4L+1$ colors
with only 1~recoloring per update;
when $L$ is sufficiently small this gives a better result.

\section{Online conflict-free colorings} \label{se:online}
In this section we consider recolorings to be too expensive and hence forbid them all together.
In this setting, any conflict-free coloring must be a proper coloring if we allow deletions:
if any point is contained in two intervals of the same color, then deleting all other intervals
would result in a conflict. Thus we would need $n$ colors in the worst case. We therefore
only allow insertions, that is, we study the online version.
Note that Theorem~\ref{thm_lowerbound} does not cover the case~$r=0$. 
However, Abam~\etal~\cite{Abam2014} proved that there is a logarithmic lower
bound on the number of colors for the online case. For completeness we provide the proof.

\begin{theorem}[Abam~\etal~\cite{Abam2014}]
For each~$n>0$, there is a sequence of~$n$ intervals for which
any online conflict-free coloring requires at least~$\lfloor\log_2 n\rfloor +1$ colors.
\end{theorem}
\begin{proof}
We reduce the static conflict-free coloring problem for~$n$ points with respect to
intervals in ~$\reals^1$ to the online problem for coloring intervals with respect to points.
The former problem asks for a coloring of the~$n$ given points,
which we can assume to be the integers~$1,\ldots, n$,
such that for any two points~$p,q$, the subset
of consecutive points between~$p$ and~$q$ has a unique color.
Since this problem requires~$\lfloor\log_2 n\rfloor +1$
colors~\cite{even-cf-03}, this reduction proves the theorem.

We map the instance~$\mathcal{I}\{1,\ldots,n\}$ to an instance~$\mathcal{I}'$ of
the online conflict-free coloring of intervals as follows: for each~$i=1,\ldots, n$,
we insert the interval~$I_i:=[-i, i]$ at time~$i$.
See Fig.~\ref{fig:mapping} for an illustration.

\begin{figure}[h]
\begin{center}
\begin{tikzpicture}
\foreach \i in {1,...,6}{
	\node at (-8+0.5*\i, 0.4) (\i) {};
	\draw[thick] (\i) circle (0.05);
}

\node[below=1mm] at (1) {$1$};
\node[below=1mm] at (6) {$n$};

\draw[thick, arrows={-latex'}] (-4.5,0.4) --++ (1,0) ;

\foreach \h/\l in {0/1, 0.2/1.5, 0.4/2,
					0.6/2.5, 0.8/3, 1/3.5}{
	\draw[thick] (-\l,\h) -- (\l,\h);
}

\draw[thick, ->] (4.5,0) --node[right]{time} (4.5, 1);
\node at (1.3,-0.05) {$I_1$};
\node at (3.8,0.95) {$I_n$};

\node at (-6.25,1.5) {$\mathcal{I}$};
\node at (0,1.5) {$\mathcal{I}'$};
\end{tikzpicture}
\end{center}
\caption{Mapping an instance of $n$ points to one of $n$ intervals.}
\label{fig:mapping}
\end{figure}
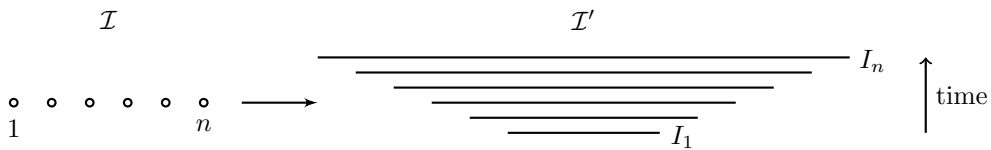

We now prove that there is an interval containing exactly the
subset~$S\subseteq \mathcal{I}$
if and only if there is a point~$q$ and a time step~$t$ such that
the point~$q$ is contained in exactly~$S':=\{I_i | p_i \in S \}$ at time~$t$
in~$\mathcal{I}'$.
Let~$i,j$ be such that~$S=\{ i,\ldots, j\}$.
We claim that the query point~$q = i-\frac{1}{2}$ at
time~$j$ is contained in exactly~$S'$; see Fig.~\ref{fig:subset}.

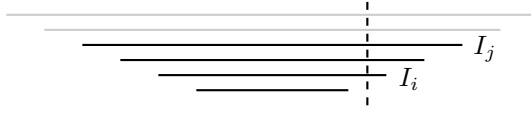
\begin{figure}[h]
\begin{center}
\begin{tikzpicture}
\foreach \h/\l in {0/1, 0.2/1.5, 0.4/2,
					0.6/2.5}{
	\draw[thick] (-\l,\h) -- (\l,\h);
}
\node at (1.8, 0.15) {$I_i$};
\node at (2.8,0.55) {$I_j$};
\foreach \h/\l in {0.8/3, 1/3.5}{
	\draw[thick, color=black!20] (-\l,\h) -- (\l,\h);
}
\draw[thick, dashed] (1.25,-0.2) -- (1.25,1.2);
\end{tikzpicture}
\end{center}
\caption{The intervals $I_i, \ldots, I_j$ that correspond to points $p_i,\ldots, p_j$.}
\label{fig:subset}
\end{figure}

First, since~$q$ is
not contained in~$I_{i-1}=[-i+1, i-1]$, and~$I_{i-1}$ contains
all intervals with lower index,~$q$ is not contained in $I_\ell$ for $\ell < i$. Moreover, since~$q$ is contained in~$I_i=[-i, i]$
and~$I_i$ is contained in any higher indexed interval,~$q$ is
also contained in those. Since we are at time~$j$,
only intervals~$I_1,\ldots, I_j$ are present.
Thus, at time~$j$, the point~$q$ is indeed contained in exactly the intervals of~$S'$, which concludes the proof.
\end{proof}

Notice that the intervals in the lower-bound construction are \emph{nested},
that is, any two intervals are either disjoint or one contains the other.
We show that for inputs restricted to such nested intervals a simple greedy
algorithm requires only $\lfloor\log_2 n\rfloor +1$ colors. For the case where all nested intervals cover a specific point, Abam~\etal~\cite{Abam2014} provided an algorithm that uses $O(\log n)$ colors.

We label the colors as $0,1,2,3,\ldots$ in the order in which they are introduced.
For each newly inserted interval, the greedy algorithm gives it the dummy color~$0$
if it is contained in another interval; otherwise, it gets the smallest color~$c \geq 1$
that keeps the coloring conflict-free. Next we prove that this greedy algorithm
with a dummy color is optimal on nested intervals.

\begin{theorem}
The greedy algorithm with a dummy color uses
at most~$\lfloor\log_2 n\rfloor +1$ colors on nested instances.
\end{theorem}
\begin{proof}
We prove the following statement:
\begin{quotation}
\noindent For any ordered set $S$ of $n$ nested intervals and
for any color $i\geq 1$ the following holds: if the greedy algorithm assigns color~$i$
to an interval $I\in S$, then $I$ contains least~$2^{i-1}-1$ other intervals from~$S$.
\end{quotation}
If this statement holds, then the algorithm indeed uses
at most~$\lfloor\log_2 n\rfloor +1$ colors.
We prove the statement by induction on the number~$n$ of
intervals.

The case $n=1$ is obvious, so
now consider the case $n>1$ and assume the statement
holds on any instance of less than $n$ intervals.
For any $1\leq j\leq n$, define $S_j := \{I_1,\ldots,I_j\}$ to be the
first $j$ intervals in $S$.
Let $i$ be the color of~$I_n$, the last interval inserted.
Let~$S(I_n)\subseteq S_{n-1}$ be the set of intervals contained in~$I_n$.
We can assume without loss of generality that~$S(I_n)$
does not contain any interval~$I_j$ using color~$i$ or higher
otherwise we can apply the induction hypothesis to~$S_j$.

Since the algorithm did not use colors~$1,\ldots, i-1$ for $I_n$,
there is a query point~$q$ contained in exactly one interval~$I_j$
of color~$i-1$. Note that $I_j \subsetneq I_n$.
Let~$S(I_j)\subseteq S$ be the
set of intervals contained in~$I_j$ (excluding $I_j$ itself).
By the induction hypothesis,~$|S(I_{j})| \geqslant 2^{i-2}-1$.
We now prove that the color of any interval
in~$S(I_n) \setminus \{ S(I_j) \cup I_{j} \}$
is not influenced by~$S(I_j) \cup  \{  I_{j} \} $.

\myclaim{%
Any interval in~$S(I_n) \setminus \{ S(I_j) \cup \{ I_{j} \} \}$
of color~$k<i$ in the full instance is also colored with~$k$ in the
instance where only~$S(I_n) \setminus \{ S(I_j) \cup \{ I_{j} \}\}$ is
inserted.%
}%
{%
We show the claim using the following invariant:
each interval of $S(I_n) \setminus \{ S(I_j) \cup \{I_{j} \}\}$ being
colored receives the same color in the restricted
instance as it receives in the full instance.
The invariant obviously holds when no interval
is inserted. Suppose now that the invariant
holds, and let~$I_{\ell}\in S(I_n) \setminus \{ S(I_j) \cup \{I_{j} \}\}$ be the next inserted interval.
Let~$k<i$ be the color of $I_\ell$ in the full instance.
If~$I_\ell \cap I_{j} = \emptyset$, the invariant obviously holds
after the insertion of~$I_\ell$ since it holds before the insertion
of~$I_\ell$ and since no interval in~$ S(I_j) \cup \{ I_{j} \}$
intersects~$I_\ell$.
Suppose therefore that~$I_\ell \supset I_{j}$. In this case,~$k\neq i-1$
since we assumed there is a query point~$q\in I_j$ such that $I_j$ is the only interval
of color~$i-1$ containing~$q$.
Therefore,~$I_{j}$ has no influence on the color of~$I_\ell$.
Moreover, no interval in~$S(I_{j})$ can have color~$i-1$:
suppose there is an interval~$I_m \in S(I_j)$ of color~$i-1$.
Then, for any point~$p\in I_m$, there must be an interval of unique
color smaller than~$i-1$, and therefore one of these intervals,
say~$I'_m$ colored with~$i'<i-1$, must contain~$I_m$. However,
in that case, the color~$i'$ was available when coloring~$I_m$
which contradicts the greedy choice. Thus, no interval in~$S(I_j)$
has color~$i-1$ and hence the color of~$I_\ell$ is only determined
by intervals in~$S(I_n) \setminus \{ S(I_j) \cup I_{j} \}$ and
therefore the invariant is maintained for~$I_\ell$ too.%
}

We now have two cases.
\begin{enumerate}
\item By removing~$S(I_j)\cup \{I_{j}\}$, the
interval~$I_n$ still cannot use color~$i-1$ (by the claim, all
the other intervals keep the same colors).
Then there is an interval $I_\ell$ not in~$S(I_j)\cup \{I_{j}\}$,
that uses color~$i-1$ and is contained in~$I_n$, which
by the induction hypothesis contains at
least~$2^{i-2}-1$ other intervals. Since $I_\ell$ cannot contain~$I_{j}$---this follows
from the definition of $I_j$---the
intervals contained in $I_j$ and the ones contained in $I_\ell$
are distinct.
Hence,~$I_n$ contains at least~$(2^{i-2}-1)+(2^{i-2}-1)+2=2^{i-1}$ intervals.
\item By removing~$S(I_j)\cup \{I_{j}\}$, the
interval~$I_n$ can use color~$i-1$ (by the claim, all
the other intervals keep the same colors).
By the induction hypothesis,~$I_n$
contains at least~$2^{i-2}-1$ intervals
in the instance without~$S(I_j)\cup \{I_{j}\}$,
see Fig.~\ref{fig:nested} for an illustration.
By adding back~$S(I_j)$ and~$I_{j}$, we have that~$I_n$ contains
at least~$(2^{i-2}-1)+(2^{i-2}-1)+1=2^{i-1}-1$ intervals.
\end{enumerate}

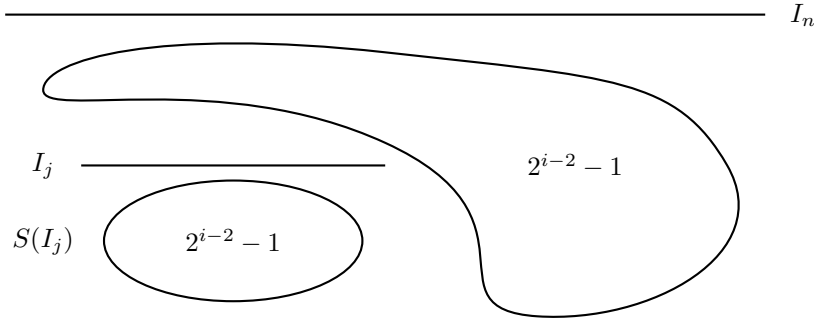
\begin{figure}[h]
\begin{center}
\begin{tikzpicture}
\draw[thick] (-5,0) -- (5,0);
\node at (5.5,0) {$I_n$};

\draw[thick] (-4,-2) -- (0,-2);
\node at (-4.5,-2) {$I_{j}$};

\draw[thick] (-2,-3) ellipse (1.7cm and 0.8cm);
\node at (-2,-3) {$2^{i-2}-1$};
\node at (-4.5,-3) {$S(I_j)$};

\draw[thick] plot [smooth cycle, tension=1.2]  coordinates {
		(-4.5,-1) (0,-0.5) (4.5,-2) (2,-4) (0,-1.7)};
\node at (2.5,-2) {$2^{i-2}-1$};
\end{tikzpicture}
\end{center}
\caption{By removing $S(I_j)$ and $I_{j}$, interval $I_n$ can use color $i-1$.}
\label{fig:nested}
\end{figure}
\end{proof}

\paragraph{Non-deterministic.}
Our main interest in this paper is in deterministic algorithms.
For the online case, we note that there is a randomized algorithm that
uses~$O(\log n)$ colors with high probability. This follows from
the general technique of Bar-Noy~\etal~\cite{barnoy-ocf-10}, who present
an online algorithm for so-called $k$-degenerate hypergraphs that uses
$O(k\log n)$ colors with high probability. A hypergraph~$H=(V,E)$
is called~\emph{$k$-degenerate} if for every subset~$V'\subseteq V$ of vertices,
and for every ordering~$v_1, \ldots, v_i$ of~$V'$, the following
holds:~$\sum_{j=1}^i \deg_{D[1,\ldots, j]}(v_j) \leqslant k|V'|$,
where~$\deg_{D[1,\ldots, j]}(v_j)$ denotes the degree of~$v_j$
in the Delaunay graph of~$H[v_1,\ldots, v_j]$ (the Delaunay
graph of a hypergraph is the subhypergraph on all vertices
containing the hyperedges of size 2).
To apply their result it suffices
to show that the hypergraph generated by intervals is 2-degenerate,
which can be done using an amortized counting of the number of
new Delaunay edges at each insertion as follows.

\begin{theorem}
Given an online sequence of intervals we can maintain a conflict free coloring of this sequence using $O(\log n)$ colors with high probability.
\end{theorem}
\begin{proof}
To apply the result of Bar-Noy \etal Let~$I_1,\ldots, I_i$ be a sequence of intervals inserted
in that order. Each time an interval is inserted,
each endpoint~$p$ is labelled (or relabelled) 0, 1, or 2, depending if
the region of the real line immediately to the right
of~$p$ is covered by zero, one, or
two or more intervals respectively. Now it suffices to notice
two things: first, the label can only increase over time, second,
each time a Delaunay edge appears, at least one endpoint label
becomes~2 (either a 1 is turned into a 2, or one of the endpoints
of the new interval is immediately labelled 2). Therefore, there
cannot be more Delaunay edges created than the number of
endpoints.
\end{proof}

\section{Kinetic conflict-free colorings} \label{se:kinetic}
\newcommand{\chain}{\mathcal{C}}

In this section we consider conflict-free colorings of a set of intervals
in $\reals^1$ whose endpoints move continuously. Note that we allow the endpoints
of an interval to move independently, that is, we allow the intervals
to expand or shrink over time but we do not allow the right endpoint
of an interval to cross the left endpoint of the same interval.
We show that by using
only three recolorings per event---an event is when two endpoints cross
each other---we can maintain a conflict-free coloring consisting of only
four colors. Our recoloring strategy is
based on the chain method discussed in the introduction.
This method uses three colors: two colors for the chain
and one dummy color. To be able to maintain the coloring in the kinetic setting without
using many recolorings, we relax the chain properties and
we allow ourselves three colors for the chain.
Next we describe the invariants we maintain on the chain and its coloring,
and we explain how to re-establish the invariants when two endpoints cross each other.
In the remainder we assume that the endpoints of the chains are in general position
except at events, and that events do not coincide (that is, we never have three
coinciding endpoints and we never have two events at the same time). These conditions can be removed by using consistent tie-breaking.

\paragraph{The chain invariants.}
Let $S$ be the set of intervals to be colored, where all interval endpoints are
distinct. (Recall that we assumed this to be the case except at event times.)
Consider a subset $\chain\subseteq S$ and order the intervals in $\chain$ according
to their left endpoint. We denote the predecessor of an interval $I\in \chain$ in this
order by $\mypred_{\chain}(I)$, and we denote its successor by $\mysucc_{\chain}(I)$.
A \emph{chain} (for $S$) is defined as a subset $\chain$ with the following three properties.
\begin{description}
\item[(C1)]
    Any interval $I\in\chain$ can intersect only two other intervals in $\chain$,
    namely $\mypred_{\chain}(I)$ and $\mysucc_{\chain}(I)$. 
\item[(C2)]
    Any interval $I \in S\setminus \chain$ is completely covered by the intervals in
    $\chain$.  
\item[(C3)] No interval $I \in \chain$ is fully contained in any other interval~$I'\in S$.
\end{description}
Now consider a set $S$ and a chain $\chain$ for~$S$. We maintain the
following \emph{color invariant}:
each interval $I\in \chain$ has a non-dummy color and this color is different
from the color of $\mysucc(I)$, and each interval in $S\setminus\chain$ has
the dummy color.
\begin{lemma}\label{lem:kin-cc-coloring}
Let $S$ be a set of intervals and $\chain$ be a chain for~$S$. Then any
coloring of $S$ satisfying the color invariant is conflict-free.
\end{lemma}
\begin{proof}
Property~(C2) implies that any point contained in at least one interval
of $S$ is contained in a chain interval. Furthermore, from property~(C1) we know
that any point $x\in\reals^1$ is contained in at most two intervals of $\chain$,
and that if $x$ is contained in two chain intervals they must be consecutive.
The color invariant guarantees that two consecutive chain intervals have different colors,
so the (at most two) chain intervals containing~$x$ provide a unique color for~$x$.
Hence, the coloring is conflict-free.
\end{proof}

\paragraph{Handling events.}
Our kinetic coloring algorithm maintains a chain~$\chain$ for $I$ and a
coloring with three colors (excluding the dummy color) satisfying the color invariant.
Later we show how to re-establish the color invariant at each event,
but first we show how to update the chain by adding at most one interval
to the chain and removing at most two.
We distinguish several cases.
\begin{figure}
\centering
\includegraphics{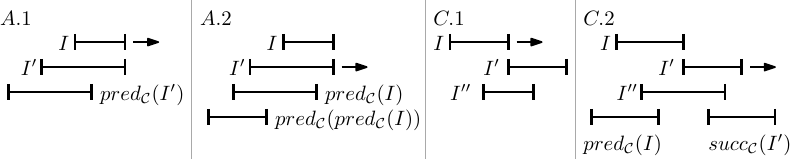}
\caption{Illustration of the different events in the KDS.}
\label{fig:kds-events}
\end{figure}
\begin{itemize}
\item \emph{Case A: The right endpoints of two intervals $I$ and $I'$ cross.} \\
    Assume without loss of generality that $I$ is shorter than~$I'$.
    We have two subcases.
    \begin{itemize}
    \item \emph{Subcase A.1: Interval $I$ is contained in $I'$ before the event.}
        In this case $I$ was not a chain interval before the event. If after the event
        $I$ is still fully covered by the chain intervals, then there is nothing to do:
        we can keep the same chain.
        Otherwise, property~(C2) is violated after the event. We now proceed as follows.
        First we add $I$ to the chain. If $I$ intersects $\mypred_{\chain}(I')$---note
        that $I'$ must be a chain interval if (C2) is violated---then we remove $I'$
        from the chain.
    \item \emph{Subcase A.2: Interval $I$ is contained in $I'$ after the event.}
        If $I$ was not a chain interval, there is nothing to do.  Otherwise
        property~(C3) no longer holds after the event, and we have to remove~$I$
        from the chain. If $I'$ is also a chain interval then this suffices.
        Otherwise we add $I'$ to the chain, and remove $\mypred_{\chain}(I)$ if $\mypred_{\chain}(\mypred_{\chain}(I))$ intersects~$I'$.
    \end{itemize}
\item \emph{Case B: The left endpoints of two intervals $I$ and $I'$ cross.} \\
    This case is symmetric to Case~A.
\item \emph{Case C: The right endpoint of an interval $I$ crosses the left endpoint of an interval~$I'$.} \\
    Again we have two subcases.
    \begin{itemize}
    \item \emph{Subcase C.1: Intervals $I$ and $I'$ start intersecting.}
        Note that properties~(C2) and (C3) still hold after the event. The only possible violation is in property~(C1), namely when both $I$ and $I'$ are chain intervals and there is a chain interval $I''$ with $\mypred_{\chain}(I'')=I$ and
        $\mysucc_{\chain}(I'')=I'$. In this case we simply remove~$I''$ from the chain.
    \item \emph{Subcase C.2: Intervals $I$ and $I'$ stop intersecting.}
        First note that this cannot violate properties (C1) and (C3).
        The only possible violation is property~(C2), namely when
        both $I$ and $I'$ are chain intervals and there is at least one non-chain interval containing the common endpoint of $I$ and $I'$ at the event.
        Of all such non-chain intervals, let $I''$ be the interval with the
        leftmost left endpoint. Note that $I''$ is not contained in any
        other interval, so we can add $I''$ to the chain without violating~(C3).
        After adding $I''$ we check if we have to remove $I$ and/or $I'$:
        if $I''$ intersects $\mypred_{\chain}(I)$ then we remove
        $I$ from the chain, and if $I''$ intersects $\mysucc_{\chain}(I')$ then we remove
        $I'$ from the chain.
    \end{itemize}
\end{itemize}
\noindent It is easy to check that in each of these cases the new chain that we generate
has the chain properties~(C1)--(C3). Next we show that each case requires at most
three recolorings and summarize the result.
\begin{lemma}\label{lem:kin-recolorings}
In each of the above cases, the changes to the chain require at most three recolorings to re-establish the color invariant.
\end{lemma}
\begin{proof}
An interval that is a non-chain interval before and after the event need not be recolored.
We have seen that we add at most one interval to the chain and remove at most two.
Removing an interval from the chain requires recoloring it to the dummy color.
Adding an interval requires giving it a color that is different from its predecessor
and its successor in the chain, which we can do since we have three colors
available for the chain. It remains to check if any other chain intervals need to be recolored.
This can only happen in case C.1, when $I$ and $I'$ are chain intervals that are not removed.
In this case it suffices to recolor $I$ with a color different from the color
of~$I'$ and from the color of~$\mypred_{\chain}(I)$. In all cases, the number
of recolorings is at most three.
\end{proof}

\begin{theorem}\label{th:kinetic}
Let $S$ be a kinetic set of intervals in~$\reals^1$. We can maintain a conflict-free
coloring for $S$ with four colors at the cost of at most three recolorings per
event, where an event is when two interval endpoints cross each other.
\end{theorem}

\begin{wrapfigure}{R}{1.5in}
\centering
\includegraphics{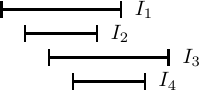}
\caption{The gadget used to show the lower bound.}
\label{fig:kds-efficiency}
\end{wrapfigure}
\paragraph{A lower bound.}
Now consider the simple scenario where the intervals are rigid---each interval keeps
the same length over time---and each interval is either stationary or moves with
unit speed to the right. Our coloring algorithm may perform recolorings whenever
two endpoints cross, which means that we do $O(n^2)$ recolorings in total.
We show that even in this simple setting, this bound is tight in the worst case
if we use at most four colors.

Consider four intervals $I_1, I_2, I_3, I_4$ where $I_i = (a_i,b_i)$, with $a_i < b_i$ as shown in Figure~\ref{fig:kds-efficiency}. Here $I_2 \subset I_1$, $I_4 \subset I_3$, the right endpoints of $I_1$ and $I_2$ are contained in $I_3 \cap I_4$, and the left endpoints of $I_3$ and $I_4$ are contained in $I_1 \cap I_2$. The exact locations of the endpoints with respect to each other is not important and we focus on the different overlap sets of the gadget. Specifically within a gadget there is a point contained in each of the following sets,
$$G_1, \ldots, G_7 :=
\{I_1\},
\{I_1, I_2\},
\{I_1, I_2, I_3, \},
\{I_1, I_2, I_3, I_4\},
\{I_1, I_3, I_4\},
\{I_3, I_4\},
\{I_3\}.$$
Based on these sets we can show that no coloring for crossing gadgets exists that provides a valid conflict-free coloring for each combination of intersection sets between the two gadgets. The proof relies on the following lemma.

\begin{lemma}\label{lem:kin-gadget}
Let $G = \{I_1, I_2, I_3, I_4\}$ and $H = \{J_1, J_2, J_3, J_4\}$ be two gadgets, with overlap sets $G_1, \ldots, G_7$ and $H_1, \ldots, H_7$ as defined above. There is no 4-coloring for $G$ and $H$
such that all sets
$\{G_1, \ldots, G_7\} \cup \{H_1, \ldots, H_7\} \cup \{ G_i \cup H_i |\  1\leq i,j \leq 7\ \}$
are conflict-free.
\end{lemma}
\begin{proof}
We can assume that not both~$I_1, I_2, I_3, I_4$
and~$J_1, J_2, J_3, J_4$ use all four colors,
otherwise~$G_4 \cup H_4 = \{I_1, I_2, I_3, I_4,J_1, J_2, J_3, J_4\}$
is not conflict-free. It is also not possible
to use at most two colors, since each gadget by itself needs to be conflict-free. Hence, suppose that there
are exactly three colors among~$I_1, I_2, I_3, I_4$
(the other case is symmetric), say two are red, one is blue,
and one is green. We define~$\col(G_i)$, respectively~$\col(H_i)$,
to be the multiset
of the colors used by the intervals in~$G_i$, respectively~$H_i$.
Then~$\col(G_4)=\{\mbox{red}, \mbox{red}, \mbox{blue}, \mbox{green} \}$
and without loss of generality,~$\col(G_2)=\{\mbox{red}, \mbox{blue} \}$
and~$\col(G_6)=\{\mbox{red}, \mbox{green} \}$.
We now have two cases.
\begin{enumerate}
\item One interval among~$J_1, J_2, J_3, J_4$
uses the fourth color, say yellow.
If~$J_1$ or~$J_2$ is yellow, then either~$\col(H_6)=\{\mbox{red}, \mbox{blue} \}$,
implying that $G_2 \cup H_6$ is not conflict-free;
or~$\col(H_6)=\{\mbox{red}, \mbox{green} \}$
implying that $G_6 \cup H_6$ is not conflict-free;
or~$\col(H_6)=\{\mbox{blue}, \mbox{green} \}$
implying that $G_4 \cup H_6$ is not conflict-free.
A similar argument holds when ~$J_3$ or~$J_4$ is yellow.
\item Two intervals among~$J_1, J_2, J_3, J_4$
use yellow. It follows that $H_4$ contains two yellow intervals and the remaining two intervals are colored either~$\{\mbox{red}, \mbox{blue} \}$,
implying that~$G_2 \cup H_4$ is not conflict-free;
or~$\{\mbox{red}, \mbox{green} \}$,
implying that~$G_6 \cup H_4$ is not conflict-free;
or~$\{\mbox{blue}, \mbox{green} \}$,
implying that~$G_4 \cup H_4$ is not conflict-free. \\
\end{enumerate}
\vspace{-2.5\baselineskip}
\end{proof}

Now we place $\Omega(n)$ of these gadgets in two groups and for simplicity assume a gadget has width of 1. The gadgets in the first group are spaced with distance 2 between them, so a gadget from the second group can fit between any two consecutive gadgets. In the second group the gadgets are spaced with distance $3n$ between them, so that all gadgets of the first group fit between them. All gadgets of the first group then move at the same speed, starting somewhere to the left of the second group and moving to the right. The gadgets of the second group remain stationary. These motions ensure that each gadget of first group will cross each gadget of the second group, generating $\Omega(n^2)$ crossing events, each of which results in at least one recoloring by Lemma~\ref{lem:kin-gadget}.

\begin{theorem}
For any $n>0$, there is a set of $8n$ intervals, each of which is either stationary or
moves with unit speed to the right, so that when coloring the intervals using four colors at least $n^2$ recolorings are required to maintain a conflict-free coloring.  
\end{theorem}

\section{Conclusions}
We introduced the fully-dynamic conflict-free coloring problem, where we want to
maintain a conflict-free coloring for a set of geometric objects under
insertions and deletions, and we presented
a number of lower bounds and algorithms for the 1-dimensional version of the problem.
There are several open problems.

First of all, there is still a gap between our upper
and lower bounds. Recall that Theorem~\ref{thm_lowerbound} gives a lower
bound in the semi-dynamic (insertion-only) case. It would be interesting to see
if a stronger lower bound can be obtained by mixing insertions and deletions.

Second, our lower bounds give a bound on the number of colors needed
as a function of the worst-case number of re-colorings
per insertion. It would also be interesting to investigate lower bounds
as a function of the amortized number of re-colorings and/or to develop
better algorithms in this setting.

Finally, an obvious area of future research is to consider the problem in two- or higher-dimensional
space. In particular the problem of maintaining a conflict-free coloring for a set of disks in the
plane under insertions and deletions---which comes much closer to the application that motivated
conflict-free colorings, namely frequency assignment in wireless networks---is of interest.
We note that very recently De Berg and Markovic~\cite{dynamic-CF-2d} obtained several results on
the two-dimensional problem.

\bibliography{dynamic-cf-coloringbibshort}

\end{document}